\documentclass[12pt,english,american]{article}

\usepackage[colorlinks=true,linkcolor=blue,citecolor=red]{hyperref}

\usepackage{amsmath,amsfonts,amssymb,amsthm,mathrsfs}

\usepackage{enumerate}
\usepackage{bm}
\usepackage{bbm}
\usepackage{verbatim}
\usepackage[usenames,dvipsnames]{color}
\usepackage{xargs} 
\usepackage{authblk}
\usepackage[square,sort,comma,numbers]{natbib}

\theoremstyle{plain}\newtheorem{theorem}{Theorem}[section]
\theoremstyle{plain}\newtheorem{lemma}[theorem]{Lemma}
\theoremstyle{plain}\newtheorem{corollary}[theorem]{Corollary}
\theoremstyle{plain}
\theoremstyle{plain}\newtheorem{proposition}[theorem]{Proposition}
\theoremstyle{definition}
\theoremstyle{remark}\newtheorem{remark}{Remark}
\theoremstyle{definition}\newtheorem{def:and:lemma}[theorem]{Definition and Lemma}

\newcommandx{\unsure}[2][1=]{\todo[linecolor=red,backgroundcolor=red!25,bordercolor=red,#1]{#2}}

\newcommand{\D}{\textnormal{d}}

\newcommand{\norm}[2][]{\big | \hspace{-0.2mm} \big | #2 \big | \hspace{-0.2mm} \big |_{#1}}
\newcommand{\snorm}[2][]{| \hspace{-0.2mm}| #2 | \hspace{-0.2mm} |_{#1}}

\newcommand{\T}[1]{}

\usepackage[colorinlistoftodos,prependcaption,textsize=tiny]{todonotes}

\newcounter{remarks}
\newcommand{\lsp}{\big \langle }
\newcommand{\rsp}{\big \rangle }
\newcommand{\rspH}{\big \rangle_{\hspace{-1mm}\mathscr H} }

\newcommand{\im}{\operatorname{Im}}
\newcommand{\re}{\operatorname{Re}}

\begin{document}

\bibliographystyle{alpha}

\title{\textbf{A note on the Fr\"ohlich dynamics in the\\ strong coupling limit}
}

\author{David Mitrouskas}

\date{\vspace{-5.5mm}\today}

\maketitle


\frenchspacing

\begin{abstract}
We revise a previous result about the Fr\"ohlich dynamics in the strong coupling limit obtained in \cite{Griesemer2017}. In the latter it was shown that the Fr\"ohlich time evolution applied to the initial state $\varphi_0 \otimes \xi_\alpha$, where $\varphi_0$ is the electron ground state of the Pekar energy functional and $\xi_\alpha$ the associated coherent state of the phonons, can be approximated by a global phase for times small compared to $\alpha^2$. In the present note we prove that a similar approximation holds for $t=O(\alpha^2)$ if one includes a nontrivial effective dynamics for the phonons that is generated by an operator proportional to $\alpha^{-2}$ and quadratic in creation and annihilation operators.
Our result implies that the electron ground state remains close to its initial state for times of order $\alpha^2$ while the phonon fluctuations around the coherent state $\xi_\alpha$ can be described by a time-dependent Bogoliubov transformation.
\end{abstract}
\vspace{0mm}
\noindent \text{MSC class:} 81Q05, 81Q15, 82C10
\\[0.5mm]
\noindent \text{Keywords:} Fr\"ohlich polaron, strong coupling limit, effective dynamics, quantum corrections

\section{Introduction and Main Result}

1.1.\ \textbf{The model.} The Fr\"ohlich polaron is a quantum model for a large polaron which describes an electron in an ionic lattice interacting with the excitations (phonons) of this lattice \cite{Froehlich1937,DevreeseA2010}. \textit{Large} refers to the assumption that the extension of the electron is much larger compared to the lattice spacing which can thus be approximated by a continuum. In this model, the energy and the dynamics of the electron and the phonons are described by the Fr\"ohlich Hamiltonian 
\begin{align}\label{eq: def froehlich ham} 
H^{\rm F}_{\rm phys,\alpha}  = p^2 \otimes 1 + 1 \otimes N + \sqrt \alpha \phi(G_x)
\vspace{1.5mm}
\end{align}
that acts  on the Hilbert space $\mathscr H = L^2(\mathbb R^3,\D x)\otimes \mathcal F$. Here $\mathcal F = \bigoplus_{n=0}^\infty  L^2(\mathbb R^3,\D k)^{\otimes_{\rm sym}^n} $ is the bosonic Fock space, $x$ and $p= -i\nabla_x$ denote the position and momentum operator of the electron, respectively, and $N$ is the number operator on $\mathcal F$. The interaction between the electron and the phonons is described by $\phi(G_x) = a( G_x ) + a^*(G_x)$ with $a(f)$ and $a^*(f)$ the usual annihilation and creation operators on $\mathcal F$ and $G_x$ the bounded multiplication operator defined for any $x\in \mathbb R^3$ by the function
\begin{align}\label{eq: definition of G}
G_x(k) = \frac{e^{ - ikx}}{ 2\pi \vert k \vert}.
\end{align}
The creation and annihilation operators satisfy the canonical commutation relations
\begin{align}\label{eq: canonical commutation relations}
[ a(f) , a^*(g) ] = \langle f,g \rangle_{L^2}, \quad  [ a(f) , a(g) ] =  [ a^*(f) , a^*(g) ] = 0\quad \ \forall \ f,g\in L^2(\mathbb R^3,\D k).
\end{align}
Finally the number $\alpha >0$ is a dimensionless coupling parameter that models the strength of the interaction. The regime $\alpha \to \infty$ is called the strong coupling limit.

By a change of units which corresponds to rescaling all lengths by a factor $\alpha^{-1}$, the Fr\"ohlich Hamiltonian $  H^{\rm F}_{\rm phys,\alpha}$ is unitarily equivalent to the operator $\alpha^{2}   H^{\rm F}_{\alpha}$ with\footnote{See \cite[Appendix A]{FrankS2014} or \cite[Appendix B]{Griesemer2017}.} 
\begin{align} \label{eq: definition Froehlich Hamiltonian in scu}
 H^{\rm F}_{\alpha} =  p^2 \otimes 1 + 1 \otimes \alpha^{-2}N + \alpha^{-1} \phi (G_x) .
\end{align}
In the analysis of the strong coupling limit it is more convenient to work in strong coupling units, i.e.\ to use $H^{\rm F}_{\alpha}$ instead of the original Fr\"ohlich Hamiltonian $H^{\rm F}_{\textnormal{phys},\alpha}$ and then consider rescaled values of energy $E = \alpha^2 E_{\rm phys}$ and time $t =  \alpha^2  t_{\rm phys}$. This explains why $t=O(\alpha^2)$ is the time scale we are interested in for the dynamics generated by $H^{\rm F}_\alpha$.

In this work we study the large $\alpha$ limit of the time evolved state $\Psi_\alpha (t) = e^{-iH^{\rm F}_\alpha t }\Psi_\alpha $ for a special initial state, namely the Pekar product state $\Psi_\alpha = \varphi_0 \otimes  \xi_\alpha $ where $\varphi_0\in H^1(\mathbb R^3,\D x)$ is the self-trapped electron ground state of the Pekar energy functional (to be defined below) and $ \xi_\alpha = W(\alpha f_0 )^* \Omega_0$ is the corresponding coherent phonon state. That is to say, $\Omega_0 = (1,0,0,...)$ is the normalized vacuum state in $\mathcal F$ and
\begin{align}\label{eq: def of Weyl operator}
W(\alpha f_0 ) = \exp\big(  a^*(\alpha f_0 ) - a( \alpha f_0 )  \big)
\end{align}
denotes the Weyl operator w.r.t. the function
\begin{align}\label{eq: def of f}
\alpha f_0(k) = \alpha \lsp \varphi_0, G_x (k) \varphi_0 \rsp_{L^2}  =  \frac{\alpha}{2\pi \vert k \vert} \int_{\mathbb R^3}  e^{-ikx} \vert \varphi_0(x)\vert^2  \, \D x . 
\end{align}
We recall that the Weyl operator is unitary and satisfies the shift relation
\begin{align}\label{eq: shift relations Weyl}
W(\alpha f_0)^* a(g) W(\alpha f_0) = a(g) + \alpha \lsp g,f_0\rsp_{L^2}
\end{align}
for any $g\in L^2(\mathbb R^3,\D k)$.

The Pekar energy functional is defined by 
\begin{align}\label{eq: Pekar energy functional}
\mathcal E^{\rm P}(\varphi ) = \int_{\mathbb R^3} \vert \nabla \varphi (x)\vert^2 \D x - \frac{1}{2}\int_{\mathbb R^3} \int_{\mathbb R^3} \frac{\vert \varphi (x)\vert^2 \vert \varphi (y)\vert^2}{\vert x-y\vert}
 \D x \D y
\end{align}
with constraint $\snorm[L^2]{\varphi} =1$. It was shown in \cite{Lieb1977} that $\mathcal E^{\rm P}(\varphi)$ admits a unique minimizer (unique up to spatial translations)
\begin{align}\label{eq: Pekar ground state}
\varphi_0 \in H^1(\mathbb R^3, \D x ) \cap \big\{\varphi  \in L^2(\mathbb R^3,\D x )\, : \,  \snorm[L^2]{\varphi }=1\big\}
\end{align}
that can be chosen positively. The minimizer further solves the Euler--Lagrange equation $(h^{\varphi_0} - \lambda ) \varphi_0 =0 $ where
\begin{align}\label{eq: def pekar hamiltonian}
h^{\varphi_0} = p^2 + V^{\varphi_0} , \quad V^{\varphi_0}(x) = - 2\re \langle G_x , f_0 \rangle_{L^2},
\end{align}
and $\lambda  = \mathcal E^{\rm P}(\varphi_0) - \snorm[L^2]{f_0}^2$. By its positivity, it follows that $\varphi_0$ is the unique ground state of the Schr\"odinger operator $h^{\varphi_0}$ and that $\lambda  = \inf \sigma(h^{\varphi_0})$ belongs to the discrete spectrum of $h^{\varphi_0}$, see\ \cite[Sec.\ 12]{ReedSimonVol4}. Introducing the orthogonal projector $Q = 1- P$ with $P = \vert\varphi_0 \rangle \langle \varphi_0 \vert$, we hence known
that $h^{\varphi_0} - \lambda$ is a positive operator when restricted to the closed subspace $\textnormal{ran}Q \subseteq L^2(\mathbb R^3,\D x)$. This allows the definition of the restricted resolvent
\begin{align}\label{eq: def of restricted resolvent}
R = Q ( h^{\varphi_0} - \lambda )^{-1} Q 
\end{align}
as a bounded operator in $L^2(\mathbb R^3,\D x)$. The fact that $R$ is independent of $\alpha$ and thus bounded uniformly as $\alpha\to \infty$ is a crucial ingredient in the analysis of the strong coupling limit of $\Psi_\alpha(t)$. In a nutshell, it ensures a separation of scales as $\alpha \to \infty$ of the different parts of the Fr\"ohlich Hamiltonian $H^{\rm F}_\alpha$ when the latter is applied to states of the form $\varphi \otimes W(\alpha f_0)^* \eta$ for suitable $\varphi \in \text{ran}Q$ and $\eta \in \mathcal F$.

That the scale separation of the different parts in $H_\alpha^{\rm F}$ allows an effective description of the Fr\"ohlich dynamics for times $t=o(\alpha^2)$ was first observed in \cite{Griesemer2017}. There it was shown that the wave function $\Psi_\alpha(t) = e^{-iH^{\rm F }_\alpha t} \varphi_0 \otimes W(\alpha f_0 )^* \Omega_0  $ remains close to its initial state up to a global phase factor, i.e.
\begin{align}\label{eq: main estimate Griesemer}
\norm[\mathscr H]{  \Psi_\alpha(t)  - e^{-i  \mathcal E^{\rm P}(\varphi_0)  t}  \varphi_0 \otimes W(\alpha f_0 )^* \Omega_0  } \le C \, \vert t\vert^{1/2} \alpha^{-1}
\end{align} 
for some $C>0$. Since the initial state is normalized to one, the upper bound is meaningful for $t \ll \alpha^2$. A similar approximation was obtained in \cite{LeopoldRSS2019} for more general initial states, namely Pekar product states in which the electron is initially trapped in the classical field produced by a given coherent state of the phonons. Modulo a global phase factor, the effective dynamics is then described by the Pekar product state  $\varphi_{\text{\tiny LP}}(t) \otimes W(\alpha f_\text{\tiny LP}(t))^*\Omega_0$ with $(\varphi_\text{\tiny LP}(t), f_\text{\tiny LP}(t))$ solving the time-dependent Landau--Pekar equations, cf. \cite[Eqn. (8)]{LeopoldRSS2019}. In fact, the effective dynamics in \eqref{eq: main estimate Griesemer} can be understood as the special case in which $(\varphi_0, f_0)$ are the stationary ground state solutions of the Landau--Pekar equations. The proof of the nonstationary problem, however, is technically more demanding as it is based on a nonlinear adiabatic theorem for the solution of the Landau--Pekar equations, see \cite[Theorem II.1]{LeopoldRSS2019}. Loosely speaking, the latter shows that the scale separation of the different parts in the Fr\"ohlich Hamiltonian remains valid on some suitable time scale also in the nonstationary case. An adiabatic theorem for the Landau--Pekar equations in one spatial dimension has been derived in \cite{Frank2017,FrankG2019}. Earlier results about the Fr\"ohlich dynamics in the strong coupling limit provide approximations for $t=o(\alpha)$ but for much more general initial Pekar product states $\varphi \otimes W(\alpha f)^*\Omega_0$ with no particular assumption about the relation between $\varphi$ and $f$, see \cite{FrankS2014,FrankG2017}. To our knowledge, there are no results available to date that provide an approximation for the Fr\"ohlich dynamics for $t=O(\alpha^2)$.

\begin{remark} The particular choice of our initial state $\Psi_\alpha = \varphi_0 \otimes \xi_\alpha$ is motivated by Pekar's approximation of the ground state energy of the Fr\"ohlich Hamiltonian \cite{Pekar1954}. Taking the expectation value of $H^{\rm F}_\alpha$ for general Pekar states $\varphi \otimes W(\alpha f)^*\Omega_0$ and minimizing over the phonon mode $f \in L^2(\mathbb R,\D k )$ leads to the Pekar functional $\mathcal E^{\rm P}(\varphi )$. That Pekar's approximations is accurate in the strong coupling limit was rigorously proved in \cite{Donsker1983} and later, using a different approach which provided in addition a quantitative error erstimate, in \cite{LiebT1997}. They showed
\begin{align}\label{eq: Pekar approximation energy}
\inf \sigma ( H^{\rm F}_\alpha ) =  \mathcal E^{\rm P}(\varphi_0 )  + o(1)
\end{align}
as $\alpha\to \infty$. The physical picture behind this result is that the electron creates a classical phonon field which in turn leads to an effective trapping of the electron. This self-trapping mechanism is described by the ground state of \eqref{eq: Pekar energy functional}. Let us also mention that the rigorous derivation of the next order contribution in \eqref{eq: Pekar approximation energy} is still an open problem that was recently solved in \cite{FrankS2019} for a model in which the Fr\"ohlich polaron is assumed to be confined to a suitably bounded region $\Lambda \subset \mathbb R^3$.
\end{remark}

\begin{remark} We note that \eqref{eq: definition Froehlich Hamiltonian in scu}, and equally \eqref{eq: def froehlich ham}, is somewhat formal since $G_x\notin L^2(\mathbb R^3,\D k)$ and hence $\phi(G_x)$ is not a densely defined operator. However, by a well-known argument that goes back to Lieb and Yamazaki \cite{LiebY1958}, the right side of \eqref{eq: definition Froehlich Hamiltonian in scu} defines a closed semi-bounded quadratic form with domain given by the form domain of $ p^2\otimes 1 + 1\otimes N$. The Hamiltonian $H^{\rm F}_{\alpha}$ is then defined as the unique self-adjoint operator associated with this quadratic form, cf. \cite[Thm. VIII.15]{ReedSimonVol1}. For the purpose of this work it is sufficient to use the form representation given in \eqref{eq: definition Froehlich Hamiltonian in scu}. Alternative approaches to define the Fr\"ohlich Hamiltonian with an explicit characterization of its domain were discussed more recently in \cite{GriesemerWuensch2016,LampartSchmidt19}.
\end{remark}
\vspace{1mm}
\noindent 1.2.\ \textbf{Effective dynamics.} Our goal is to derive an approximation similar to \eqref{eq: main estimate Griesemer} for times $t= O(\alpha^2)$. To achieve this, we compare $\Psi_\alpha(t)$ with an effective time evolution that is generated by the Hamiltonian
\begin{align}\label{eq: effective hamiltonian without Weyl operators}
H_{\alpha}^{\varphi_0} = 1 \otimes \lsp \varphi_0 , \big( H^{\rm F}_\alpha -  (\alpha^{-1} \phi(G_x )-V^{\varphi_0}) (R  \otimes 1) (\alpha^{-1} \phi(G_x )-V^{\varphi_0} )  \big) \varphi_0 \rsp_{L^2}.
\end{align}
In the following proposition we clarify the difference compared to the ansatz in \eqref{eq: main estimate Griesemer} and, more importantly, we obtain the existence of a unitary time evolution generated by $H_{\alpha}^{\varphi_0}$.

\begin{proposition}\label{prop: effective hamiltonian} For any $\alpha >0$ we have
\begin{align}\label{eq: effective hamiltonian}
 W(\alpha f_0) H^{\varphi_0}_{ \alpha }   W(\alpha f_0)^* -   \mathcal E^{\rm P}(\varphi_0)  =  1\otimes   \alpha^{-2}(N-A^{\varphi_0})
\end{align}
with the operator $ A^{\varphi_0}: \mathcal F \to \mathcal F$ defined by 
\begin{align}
 A^{\varphi_0} =  \big\langle \varphi_0 , \phi(G_x ) (R \otimes 1) \phi(G_x )   \varphi_0 \big\rangle_{L^2 } .\label{eq: def of K} 
\end{align}
Moreover, $\mathscr D(N) \subseteq \mathscr D(N-A^{\varphi_0})$ and $N-A^{\varphi_0}$ is essentially self-adjoint on $\mathcal F$. (We denote its closure again by $N-A^{\varphi_0}$.)
\end{proposition}
We prove this proposition in Section \hyperref[sec: 2.4]{2.4}. By unitarity of the Weyl operator, it follows that $H_{\alpha}^{\varphi_0}$ is self-adjoint on $\mathscr H$ and thus $\exp ( -iH_{\alpha}^{\varphi_0} t )$ defines a unitary time evolution.

Let us emphasize that the effective Hamiltonian acts nontrivially only on the phonons. This implies in particular that the time evolved state $\exp ( -iH_{\alpha}^{\varphi_0} t ) \varphi_0 \otimes \xi_\alpha$ is still an exact product. Because of the operator $A^{\varphi_0}$ in \eqref{eq: effective hamiltonian}, however, the coherent state structure of the initial state $\xi_\alpha$ is not conserved. In this regard, our effective dynamics is different compared to the known results discussed in the previous section.

\begin{remark} As a motivation of our ansatz in \eqref{eq: effective hamiltonian without Weyl operators} let us mention its analogy to the well-known second order perturbation formula 
\begin{align}
E_\varepsilon = \big\langle u_0, \big( H_\varepsilon -\varepsilon V R_0  \varepsilon V \big)  u_0 \big\rangle + O(\varepsilon^3)\quad ( \varepsilon \ll 1)
\end{align}
for the nondegenerate ground state energy $E_\varepsilon$ of a suitable Hamiltonian $H_\varepsilon  = H_0 + \varepsilon V$ by means of the ground state vector $u_0$ of $H_0$ and the reduced resolvent $R_0 = (1- \vert u_0 \rangle \langle u_0 \vert ) (H_0 - \langle u_0, H_0 u_0 \rangle)^{-1} (1- \vert u_0 \rangle \langle u_0 \vert )$. Despite this analogy, we emphasize that the expectation value in \eqref{eq: effective hamiltonian without Weyl operators} is taken only w.r.t. to the electron wave function $\varphi_0 \in L^2(\mathbb R^3,\D x)$ and not w.r.t.\ to the full Pekar product $\varphi_0 \otimes \xi_\alpha$. The reason why the expectation value w.r.t.\ $\varphi_0 \otimes \xi_\alpha$ would not lead to a good ansatz for the effective dynamics is the appearance of the factor $\alpha^{-2}$ in front of the number operator $N$.
\end{remark}

\noindent 1.3.\ \textbf{Main results.} We are now ready to state our main results.  

\begin{theorem}\label{theorem: main theorem} Let $\varphi_0 \in H^1(\mathbb R^3, \D x)$ be the unique minimizer of the Pekar functional \eqref{eq: Pekar energy functional} with $\snorm[L^2]{\varphi_0 } = 1$ and let $f_0 \in L^2(\mathbb R^3, \D k)$ be defined as in \eqref{eq: def of f}. Let further $\eta_0 \in \mathcal F$ satisfy $\snorm[\mathcal F]{\eta_0}=1$ and $\sup_{\alpha>0} \snorm[\mathcal F]{(N+1)^{5/2}\eta_0}<\infty$. Then there are constants $c,C>0$ such that
\begin{align} \label{eq: main estimate}
 \norm[\mathscr H]{ \big( e^{-iH^{\rm F }_\alpha t} - e^{-i      H^{\varphi_0}_{ \alpha } t }  \big) \varphi_0 \otimes   W(\alpha  f_0  )^* \eta_0 } \le  C \alpha^{ -1} \exp ( c  \vert t \vert \alpha^{-2} ) 
\end{align}
for all $t\in \mathbb R$ and $\alpha >0$.
\end{theorem} 

Since the initial state is normalized to one, the approximation is accurate for $t=O(\alpha^2)$ (indeed, it is accurate for $t \ll \alpha^2 \ln \alpha$). As a direct consequence of \eqref{eq: main estimate} together with $[e^{-i  H_{\alpha}^{\varphi_0} t} , P\otimes 1 ]=0$,
we obtain the following statement that shows that the reduced density of the electron remains approximately constant.

\begin{corollary}\label{cor: reduced density} Under the same assumptions as in Theorem \ref{theorem: main theorem} there exist constants $c,C>0$ such that 
\begin{align}
\textnormal{Tr}_{L^2}\Big\vert  \textnormal{Tr}_{\mathcal F}  \big\vert \Psi_{\alpha}(t) \big\rangle \big\langle \Psi_{\alpha}  (t) \big\vert    -  \big \vert \varphi_0 \big \rangle \big \langle   \varphi_0 \big\vert   \Big\vert \le  C \alpha^{-1} \exp ( c \vert t \vert \alpha^{-2} )
\end{align}
with $\Psi_{\alpha}(t) = e^{-iH^{\rm F}_\alpha t} \varphi_0 \otimes W(\alpha f_0)^* \eta_0$.
\end{corollary}

Theorem \ref{theorem: main theorem} shows that on the time scale $t=O(\alpha^2)$ it is important to include the creation and annihilation of noncoherent phonons in the effective time evolution. In earlier findings which provided approximations for $t=o(\alpha)$ \cite{FrankS2014,FrankG2017} and $t=o(\alpha^2)$ \cite{Griesemer2017,LeopoldRSS2019}, respectively, it was not necessary to take such noncoherent phonons into account as the effective dynamics was still described by exact Pekar product states. In our next corollary, we use the fact that the operator $N-A^{\varphi_0}$ is quadratic in creation and annihilation operators in order to describe the fluctuations around the coherent phonons by means of a time-dependent Bogoliubov transformation. 

To make the last statement precise we need to introduce some well-known notions related to the Bogoliubov transformation. The generalized annihilation and creation operators are defined by  $A(F) = a(f) + a^*(g)$ and $A^*(F)= a^*(f) + a(g)$, respectively, for any $F = f\oplus Jg \in L^2(\mathbb R^3,\D k)\oplus L^2(\mathbb R^3,\D k)$ where $J$ denotes the complex conjugation map $(J g)(x) = \overline{g(x)}$. A bounded invertible map $\mathcal V$ on $L^2(\mathbb R^3,\D k)\oplus L^2(\mathbb R^3, \D k)$ is called a Bogoliubov map if it satisfies
\begin{align}\label{eq: generalized annihilation operator}
A^*(\mathcal V F) =  A(\mathcal V \mathcal J F) , \quad \big[ A(\mathcal V F), A^*(\mathcal V G) \big] = \lsp F,\mathcal S G\rsp_{L^2\oplus L^2}
\end{align}
for all $F,G \in  L^2(\mathbb R^3,\D k)\oplus L^2(\mathbb R^3,\D k)$ where
\begin{align}
\mathcal J = \begin{pmatrix} 0 & J\\
J & 0
\end{pmatrix} , \quad \mathcal S = \begin{pmatrix} 1 & 0 \\
0 & -1
\end{pmatrix}.
\end{align}
In case that the Bogoliubov map $\mathcal V$ is a Hilbert--Schmidt operator, i.e.\ if $\mathcal V^*\mathcal V$ is trace class, it can be implemented as a unitary operator on $\mathcal F$. This is the content of the Shale--Stinespring condition which states that there  exists a unitary operator $U_{\mathcal V} : \mathcal F \to \mathcal F$ such that
\begin{align}\label{eq: unitary implementation}
U_{\mathcal V} A(F) U_{\mathcal V}^* = A( \mathcal VF)
\end{align}
for any $F\in L^2(\mathbb R^3, \D k)\oplus L^2(\mathbb R^3,\D k)$ if and only if $\text{Tr} \mathcal V^*\mathcal V <\infty$, see e.g. \cite[Thm. 9.5]{SolovejLN2007}. We call the operator $U_{\mathcal V}$ the Bogoliubov transformation associated with the Bogoliubov map $\mathcal V$. Finally we need the concept of (pure bosonic) quasi-free states in $\mathcal F$. A quasi-free state $\eta\in \mathcal F$ is defined by the property that there is a Bogoliubov map $\mathcal V_\eta$ such that $\eta$ can be written as the transformed vacuum $\eta = U_{\mathcal V_\eta}\Omega_0$ (in particular, $\Omega_0$ is quasi-free). For a detailed introduction to Bogoliubov transformations and quasi-free states, we refer to \cite[Sec.\ 9 and 10]{SolovejLN2007}.

Our next goal is to show that the dynamics of the noncoherent phonons in $\Psi_\alpha(t) = e^{-iH^{\rm F}_\alpha t }\varphi_0 \otimes W(\alpha f_0)^* \eta_0 $ can be described by a time-dependent Bogoliubov transformation $U_{\mathcal V_\alpha(t)}$ associated with the Bogoliubov map 
\begin{align}\label{eq: def of time dep bog map}
\mathcal V_\alpha(t) = \exp\left[-\frac{it}{\alpha^2} \, \begin{pmatrix}
1- \mathcal G   & \mathcal K \\ 
-\overline{ \mathcal K } &  -1 + \overline{ \mathcal G } 
\end{pmatrix} \right]
\mathcal V_{\alpha}(0) , \quad \mathcal V_{\alpha}(0)  = \begin{pmatrix} 1 & 0 \\
0 & 1
\end{pmatrix},
\end{align}
where $\mathcal K$, $\mathcal G $ denote integral operators in   $L^2(\mathbb R^3,\D k)$ defined by the kernels
\begin{align}\label{eq: definition of K(k,l)}
\mathcal K(k,l) & =  (2\pi \vert k \vert )^{-1} (2\pi \vert l \vert )^{-1} \big\{ \lsp \varphi_0, e^{- ikx} R e^{- ilx} \varphi_0 \rsp_{L^2}  +  \lsp \varphi_0, e^{- i l x} R e^{- i k x} \varphi_0 \rsp_{L^2}\big\}  ,\\[1mm]
\label{eq: definition of G(k,l)} \mathcal G(k,l) & =  (2\pi \vert k \vert )^{-1} (2\pi \vert l \vert )^{-1} \big\{ \lsp \varphi_0, e^{+  ikx} R e^{- ilx} \varphi_0 \rsp_{L^2}  +  \lsp \varphi_0, e^{- i  k  x} R e^{ + il  x} \varphi_0 \rsp_{L^2}\big\} ,
\end{align}
and where $\overline{\mathcal K  }$, $\overline{\mathcal G }$ are to be understood as the integral operators with kernels $\overline{\mathcal K  }(k,l) = \overline{\mathcal K  (k,l)}= \mathcal K(-k,-l)$ and $\overline{\mathcal G  }(k,l) = \overline{\mathcal G  (k,l)} = \mathcal G(l,k)$, respectively.

\begin{corollary}\label{cor: main estimate bogoliubov trafo} Under the same assumptions as in Theorem \ref{theorem: main theorem} with the additional requirement that $\eta_0 \in \mathcal F$ is quasi-free, there exist constants $c,C>0$ such that 
\begin{align}
\textnormal{Tr}_{\mathcal F}\Big\vert  \textnormal{Tr}_{  L^2}  \big\vert W(\alpha f_0) \Psi_{\alpha}(t) \big\rangle \big\langle W(\alpha f_0)  \Psi_{\alpha}  (t) \big \vert    -  \big \vert U_{\mathcal V_\alpha(t)} \eta_0 \big \rangle \big \langle   U_{\mathcal V_\alpha(t)}  \eta_0 \big \vert \Big\vert \le  C \alpha^{-1} \exp ( c \vert t \vert \alpha^{-2} )
\end{align}
with $\Psi_{\alpha}(t) = e^{-iH^{\rm F}_\alpha t} \varphi_0 \otimes W(\alpha f_0)^* \eta_0$ and $U_{\mathcal V_\alpha(t)}$ the Bogoliubov transformation associated with the time-dependent Bogoliubov map $\mathcal V_\alpha(t)$ defined in \eqref{eq: def of time dep bog map}.
\end{corollary}

The remainder of this note is organized as follows. We conclude section one with a short remark about the notation and a sketch of the proof of Theorem \ref{theorem: main theorem}. In the second section we begin by stating two preliminary lemmas which are useful for the proof of Theorem \ref{theorem: main theorem}. The latter is given in Section \hyperref[sec: 2.2]{2.2} whereas the preliminary lemmas are proved in Section \hyperref[sec: 2.3]{2.3}.  Finally we prove Propostion \ref{prop: effective hamiltonian} together with Corollaries \ref{cor: reduced density} and \ref{cor: main estimate bogoliubov trafo} in Section \hyperref[sec: 2.4]{2.4}.\medskip

\noindent 1.4. \textbf{Notation.} From now on, we omit the tensor product with the identity in operators of the form $h^{\varphi_0} = h^{\varphi_0}\otimes 1$ and $N = 1 \otimes N$. Moreover we make use of the abbreviation
\begin{align}
\delta G_x = G_x - f_0,
\end{align}
with $f_0$ defined as in \eqref{eq: def of f} and by $\varphi_0\in H^1(\mathbb R^3,\D x)$ we always denote the ground state of the Pekar functional \eqref{eq: Pekar energy functional} satisfying $\snorm[L^2]{\varphi_0}=1$. The letter $C$ is used for positive constants that are independent of $t$ and $\alpha$. The exact value of $C$ may vary from line to line.\medskip

\noindent 1.5.\ \textbf{Sketch of the proof.} The proof of Theorem \ref{theorem: main theorem} is motivated mainly by the proof of inequality \eqref{eq: main estimate Griesemer} given in \cite{Griesemer2017}. To demonstrate our main idea it is instructive to start with a sketch of the derivation of \eqref{eq: main estimate Griesemer} (in slightly different way compared to \cite{Griesemer2017}). To this end, we use the shift relation \eqref{eq: shift relations Weyl} to verify
\begin{align}
W(\alpha f_0 ) H^{\rm F}_\alpha W(\alpha f_0 )^* -  \mathcal E^{\rm P}(\varphi_0) = h^{\varphi_0} - \lambda + \alpha^{-2} N + \alpha^{-1} \phi(\delta G_x).
\end{align}
With $ W(\alpha f_0 ) e^{ -iH^{\rm F}_\alpha t  }  W(\alpha f_0 )^*  =  \exp(-i W(\alpha f_0 ) H^{\rm F}_\alpha   W(\alpha f_0 )^* t )$ and by Duhamel's principle, one then obtains 
\begin{align}
& \norm[\mathscr H]{  \big( e^{-i H^{\rm F}_\alpha  t }   - e^{-i  \mathcal E^{\rm P}(\varphi_0)  t} \big) \varphi_0 \otimes W(\alpha f_0)^* \Omega_0   }^2  \nonumber \\[1mm]
& \hspace{1cm} = - 2 \alpha^{-1} \re  \ \int_0^t  i \lsp  e^{ -  i   ( h^{\varphi_0} - \lambda  + \alpha^{-2} N + \alpha^{-1} \phi(\delta G_x) )  s }     \varphi_0 \otimes \Omega_0 ,  Q  \phi( \delta G_x)  \varphi_0 \otimes \Omega_0 \rsp_{\mathscr H} \D s . \label{eq: proof sketch 01}
\end{align}
Note that we further used $(h-\lambda)\varphi_0 =0$ and $\phi( \delta G_x)  \varphi_0 \otimes \Omega_0 =  Q \phi( \delta G_x)  \varphi_0 \otimes \Omega_0 $ which holds because of $ \langle \varphi_0, \delta G_x \varphi_0 \rangle_{L^2} =0$ (recall $P= \vert \varphi_0 \rangle \langle \varphi_0\vert$ and $Q=1-P$). A rough estimate of the right side would now lead to an upper bound proportional to $\vert t \vert \alpha^{-1}$. The reason why the right side behaves actually better than this is a phase inside the integral which oscillates with nonzero ($\alpha$-independent) frequency.\footnote{One should think of the improved $t$-dependence in $\int_0^t i e^{i b s} \D s = b^{-1} (e^{ i b t}-1)$ compared to $\int_0^t 1 \D s =t$.} To take advantage of this phase we rewrite the integrand as
\begin{align}
  \lsp  e^{i(h^{\varphi_0 } -\lambda) s  }   e^{ -  i   ( h^{\varphi_0} - \lambda + \alpha^{-2} N + \alpha^{-1} \phi(\delta G_x) )  s }       \varphi_0 \otimes \Omega_0  ,  \big( \frac{d}{ds}e^{i(h^{\varphi_0 } -\lambda) s } R \big)  \phi( \delta G_x  ) \varphi_0 \otimes  \Omega_0 \rsp_{\mathscr H}
\end{align}
and then integrates by parts. This leads to a perturbation like expansion of \eqref{eq: proof sketch 01} which among other contributions (e.g. the boundary terms which are of order $\alpha^{-1}$) includes the term
\begin{align}
2 \alpha^{-2} \re \int_0^t i   \lsp  e^{ -  i   ( h^{\varphi_0} - \lambda  + \alpha^{-2} N + \alpha^{-1} \phi(\delta G_x) )  s }       \varphi_0 \otimes \Omega _0 ,  \phi(  \delta G_x )  R \phi(  \delta G_x )  \varphi_0 \otimes  \Omega_0 \rsp_{\mathscr H}\, \D s. \label{eq: proof sketch line 02}
\end{align}
Apart from some technical difficulties being related to $G_x\notin L^2(\mathbb R^3, \D k)$, one then applies the estimate (here we use that $R$ is uniformly bounded)
\begin{align}
\big \vert  \lsp  e^{ -  i   ( h^{\varphi_0} - \lambda + \alpha^{-2} N + \alpha^{-1} \phi(\delta G_x) )  s }       \varphi_0 \otimes \Omega_0  ,  \phi(  \delta G_x )  R \phi(  \delta G_x )  \varphi_0 \otimes  \Omega_0 \rsp_{\mathscr H}\vert \le C \norm[\mathcal F]{(N+1) \Omega_0}
\end{align}
in order to arrive at $\vert \eqref{eq: proof sketch line 02} \vert \le C \alpha^{-2}\vert t \vert$. This bound is indeed the reason why \eqref{eq: main estimate Griesemer} is limited to $t=o(\alpha^2)$. Our idea to improve upon this is to use the oscillating phase in \eqref{eq: proof sketch line 02} a second time. Inserting the identity $1= P+Q$ on the left of $\phi(\delta G_x)$ we obtain two contributions,
\begin{align*} 
(\ref{eq: proof sketch line 02}.\textnormal{a}) & =  2\alpha^{-2} \int_0^t  \re i \lsp    e^{ -  i   ( h^{\varphi_0} - \lambda) + \alpha^{-2} N + \alpha^{-1} \phi(\delta G_x) )  s }    \varphi_0 \otimes \Omega_0,  Q \phi(   \delta G_x )  R \phi(  \delta G_x  )  \varphi_0 \otimes \Omega_0 \rsp_{\mathscr H}\, \D s,\\
(\ref{eq: proof sketch line 02}.\textnormal{b}) & =  2\alpha^{-2} \int_0^t \re i \lsp    e^{ -  i   ( h^{\varphi_0} - \lambda) + \alpha^{-2} N + \alpha^{-1} \phi(\delta G_x) )  s }    \varphi_0 \otimes \Omega_0 , P \phi(  \delta G_x )   R \phi(  \delta G_x )   \varphi_0 \otimes \Omega_0 \rsp_{\mathscr H}\, \D s.
\end{align*}
In the first one we can proceed similarly as before and improve the bound by partial integration to $\vert (\ref{eq: proof sketch line 02}.\textnormal{a}) \vert \le C( \alpha^{-2} + \vert t\vert \alpha^{-3})$. In the second line, however, the partial integration is not applicable since $(h^{\varphi_0}-\lambda)P=0$. In other words, there is no fast oscillating phase in this term and thus (\ref{eq: proof sketch line 02}.\textnormal{b}) seems to be really of order $\vert t \vert \alpha^{-2}$. To avoid this term in the first place we include the operator $W(\alpha f_0)^* \alpha^{-2} (N-A^{\varphi_0}) W(\alpha f_0 )$ into the effective dynamics, see \eqref{eq: effective hamiltonian}. Starting over again with the new effective dynamics we now obtain an additional term in the first-order Duhamel expansion which cancels exactly the contribution from (\ref{eq: proof sketch line 02}.\textnormal{b}), cf.\ \eqref{eq: def of h(t,alpha)} and \eqref{eq: second order P term}. Because of the nontrivial dynamics of the phonons we now have to take into account the number of excitations in the effective time evolution. Using a Gronwall argument, this is shown to be bounded by a constant times $\exp(c\vert t \vert \alpha^{-2})$ which leads to the exponential factor in \eqref{eq: main estimate}. This already explains much of our proof and aside from the technical details, it would lead to an upper bound in \eqref{eq: main estimate} that is proportional to $\alpha^{-1/2} \exp(\vert t \vert \alpha^{-2})$. By a third partial integration we can improve the accuracy of this upper bound further and finally arrive at the stated bound in \eqref{eq: main estimate}.

\begin{remark}\label{eq: remark about tracer particle} The described idea of improving the approximation to longer times by changing the effective Hamiltonian as in \eqref{eq: effective hamiltonian} was similarly used also in \cite{JeblickMPP2017,JeblickMP2018}. These works treat very different models, namely the dynamics of a single tracer particle resp.\ two tracer particles interacting with an ideal Fermi gas in the high density limit. The used approximations and the proofs of their accuracy, however, are completely analogous to the one we apply to the Fr\"ohlich Polaron. The scale separation in these models comes from the large momenta of the gas modes that are close to the Fermi surface (for the ideal Fermi gas, high density is equivalent to a large Fermi momentum).
\end{remark}
\vspace{-4mm}

\section{Proofs}

\noindent \label{sec 2.1}2.1.\ \textbf{Preliminary Lemmas.} Before we start with the proof of Theorem \ref{theorem: main theorem}, let us state two lemmas with several helpful estimates. Their proofs are postponed to Section \hyperref[Sec: remaining proofs]{2.3}.
\allowdisplaybreaks
\begin{lemma}\label{Lemma: a priori estimates} Let $P=\vert \varphi_0 \rangle \langle \varphi_0 \vert$ and $R$ as defined in \eqref{eq: def of restricted resolvent}. There is a constant $C>0$ such that for any $\Psi = \varphi_0 \otimes \eta \in \mathscr H$ with $\eta \in \mathscr D(N^{5/2})$, the following bounds hold.
\begin{align}
\norm[\mathscr H]{  R  \phi (\delta G_x  )  \Psi  } + \norm[\mathscr H]{ R \big [N, \phi(\delta G_x ) \big]   \Psi  }  & \le C   \norm[\mathcal F ]{ (  N + 1 )^{ 1 / 2 }  \eta }, \label{eq: a priori bound 1}\\[3.5mm]
 \norm[\mathscr H]{  R \phi (\delta G_x )  R \phi (\delta G_x )  \Psi  } \nonumber \\[1mm]
+ \norm[\mathscr H]{  R \big[N, \phi(\delta G_x  ) R  \phi(\delta G_x ) \big]  \Psi  }   & \le C   \norm[\mathcal F ]{ (  N + 1 ) \eta }, \label{eq: a priori bound 2}\\[3.5mm]
\norm[\mathscr H]{  R   \phi (\delta G_x )  R  \phi (\delta G_x )  R  \phi (\delta G_x )  \Psi  } & \nonumber \\[1mm] 
+ \norm[\mathscr H]{  P   \phi (\delta G_x )  R  \phi (\delta G_x )  R  \phi (\delta G_x )  \Psi  } & \nonumber \\[1mm] 
+ \norm[\mathscr H]{R \phi(\delta G_x) P  \phi(\delta G_x) R \phi(\delta G_x)  \Psi} &  \nonumber\\[1mm]
+ \norm[\mathscr H]{  R \big[   \phi (\delta G_x )  R  \phi (\delta G_x )  R  \phi (\delta G_x  ) \big]  \Psi  } & \le C   \norm[\mathcal F ]{ (  N + 1 )^{3 / 2 }  \eta }, \label{eq: a priori bound 3}\\[3.5mm]
\norm[\mathscr H]{  R \phi(\delta G_x) R \phi(\delta G_x) P  \phi(\delta G_x) R \phi(\delta G_x)   \Psi  }   & \le C   \norm[\mathcal F  ]{ (  N + 1 )^{ 2 }  \eta }, \label{eq: a priori bound 4}\\[3.5mm]
\norm[\mathscr H]{  R \phi(\delta G_x)R \phi(\delta G_x)  R \phi(\delta G_x)  P  \phi(\delta G_x) R \phi(\delta G_x)  \Psi  }   & \le C   \norm[\mathcal F  ]{ (  N + 1 )^{ 5/2 }  \eta  } \label{eq: a priori bound 5}.
\end{align}
Moreover for $\Phi \in H^1(\mathbb R^3,\D x) \otimes \mathcal F$ we have  
\begin{align} 
  \big\vert \big\langle \Phi,  \phi( \delta G_x )R \phi( \delta G_x  ) R \phi( \delta G_x  )  R \phi( \delta G_x ) \Psi  \big\rangle_{\mathscr H}  \big\vert  & \le C \big( 1 + \norm[\mathscr H]{p\Phi }\big) \norm[\mathcal F ]{ (  N+1)^{2}  \eta}. \label{eq: a priori bound 6}
\end{align}
\end{lemma}
\begin{lemma} \label{Lemma: a priori estimates 2} Let $\eta \in   \mathscr D (N^{5/2})$ with $\vert \vert \eta \vert \vert_{\mathcal F}=1$ and $\sup_{\alpha >0}\snorm[\mathcal F ] { (N+1)^{5/2} \eta_0 } < \infty$. Then there are constants $c,C>0$ such that
\begin{align}\label{eq: time dependent a priori estimates 1}
\sum_{j=0}^5 \norm[\mathcal F]{  (  N+1)^{j/2} \exp( -i \alpha^{-2}(N-A^{\varphi_0 }) t ) \eta }^2 &  \le C \exp(c \vert t \vert \alpha^{-2} ), \\[0mm]
\label{eq: time dependent a priori estimates 2}
\norm[\mathscr H ]{  p\, e^{-i H^{\rm F}_\alpha t}  \varphi_0 \otimes  W(\alpha f_0)^* \eta } &  \le C 
\end{align}
for all $t\in \mathbb R$ and $\alpha>0$.
\end{lemma}

\noindent \label{sec: 2.2}2.2.\ \textbf{Proof of Theorem \ref{theorem: main theorem}.}\label{Sec: proof of main theorem} We recall the relations
\begin{align}
W(\alpha f_0 ) H^{\rm F}_\alpha W(\alpha f_0 )^* - \mathcal E^{\rm P}(\varphi_0) & = h^{\varphi_0} -\lambda   + \alpha^{-2}N  + \alpha^{-1}\phi(\delta G_x  ),\label{eq: com relation for HF}\\[1mm]
W(\alpha f_0 )  H^{\varphi_0 }_\alpha W(\alpha f_0 )^* - \mathcal E^{\rm P}(\varphi_0) & = \alpha^{-2} (N-A^{\varphi_0}),\label{eq: com relation for Hvarphi}
\end{align}
which are verified by the commutation relations
\begin{align}
 W(\alpha f_0 ) \alpha^{-2} N   W(\alpha f_0 )^* & = \alpha^{-2} N - \alpha^{-1} \phi(f_0) + \snorm[L^2]{f_0}^2,\label{eq: com relations 1}\\[2mm]
 W(\alpha f_0  ) \alpha^{-1} \phi(  G_x   )  W(\alpha f_0 )^*  & = \alpha^{-1}\phi( G_x  ) +  V^{\varphi_0}  , \label{eq: com relations 2}
\end{align}
which in turn are easily obtained via \eqref{eq: shift relations Weyl}. Using the unitarity of the Weyl operator we thus shall estimate
\begin{align}
& \norm[\mathscr H]{ \big ( e^{-iH^{\rm F}_\alpha  t  }   - e^{-i  H^{\varphi_0 }_\alpha t} \big) \varphi_0 \otimes  W(\alpha f_0 )^* \eta_0 } \nonumber \\[2mm]
& \hspace{2cm}  = \norm[\mathscr H]{ \big( e^{- i (  h^{\varphi_0} -\lambda   + \alpha^{-2}N  + \alpha^{-1}\phi(\delta G_x  ) ) t   }   -  e^{-i \alpha^{-2} (N-A^{\varphi_0})  t} \big) \varphi_0 \otimes \eta_0 }.
\end{align}
For notational convenience let us abbreviate
\begin{align}
\psi_\alpha (t ) = e^{- i (  h^{\varphi_0} -\lambda   + \alpha^{-2}N  + \alpha^{-1}\phi(\delta G_x  ) ) t   }   \varphi_0 \otimes \eta_0, \quad   \xi_\alpha (t)  =     \varphi_0 \otimes  e^{-i \alpha^{-2} (N-A^{\varphi_0})  t}  \eta_0  .
\end{align}
Application of Duhamel's principle then leads to
\begin{align}
\norm[\mathscr H]{  \psi_\alpha (t )  - \xi_\alpha (t )  }^2   =  2\re f_\alpha (t)+ 2 \re g_\alpha  (t) \label{eq: Duhamel expansion}
\end{align}
with
\begin{align}
f_\alpha  (t) & = - i \alpha^{-1}\int_0^t \lsp \psi_\alpha (s )  ,   \phi( \delta G_x )    \xi_\alpha (s)    \rspH \,  \D s   \label{eq: def of f(t,alpha)},\\
g_\alpha (t)  & = - i   \alpha^{-2} \int_0^t   \lsp  \psi_\alpha (s ) ,   P  \phi( \delta G_x  )  R \phi(\delta  G_x )   \xi_\alpha (s)   \rspH \, \D s   . \label{eq: def of h(t,alpha)}
\end{align}
Note that here we have used  $[N-A^{\varphi_0} , P ]=0$, $P\xi_\alpha (s) = \xi_\alpha (s)$ and $(h^{\varphi_0} - \lambda ) P = 0$. With $1= P + Q$ and $\langle \varphi_0, \delta G_x \varphi_0 \rangle_{L^2} = 0$ one further obtains 
\begin{align} 
f_\alpha (t)  & = - i\alpha^{-1}  \int_0^t     \lsp  \psi_\alpha(s)   ,      Q  \phi( \delta G_x  )  \xi_\alpha(s)  \rspH \, \D s . \label{eq: Duhamel expansion A(t) new} 
\end{align}
In the first part of the proof we do three partial integrations w.r.t.\ the time variable $s$. This leads to a perturbation like expansion of \eqref{eq: Duhamel expansion A(t) new} into different contributions. In particular, after the first partial integration, we obtain one term that equals $-g_\alpha ( t )$. Since this term would contribute an error of order $\vert t\vert \alpha^{-2}$, it is crucial that we included the second order correction in the effective dynamics.  All remaining contributions will be estimated separately in the second part of the proof and finally lead to the error in \eqref{eq: main estimate}.

To prepare the first partial integration we use the restricted resolvent $R = Q ( h^{\varphi_0} - \lambda )^{-1} Q $ in order to write
\begin{align}\label{eq: prepare for partial integration}
f_\alpha (t)  =  - \alpha^{-1}  \int_0^t    \lsp  e^{ i ( h^{\varphi_0} -\lambda ) s  }  \psi_\alpha(s)  ,     \big( \frac{d}{ds} e^{ i (   h^{\varphi_0}  - \lambda ) s } R  \big)  \phi ( \delta G_x )   \xi_\alpha(s)   \rsp_{\mathscr H}  \, \D s.
\end{align}
Using
\begin{align}
\frac{d}{ds} e^{ i ( h^{\varphi_0} -\lambda  ) s  }  \psi_\alpha(s)    & = - i  e^{ i ( h^{\varphi_0} - \lambda ) s  }   \big( \alpha^{-2}N + \alpha^{-1}\phi( \delta G_x  ) \big) \psi_\alpha(s)  ,\\[1mm]
\frac{d}{ds}   \xi_\alpha(s)    & = -i  \alpha^{-2} ( N  - A^{\varphi_0}  )    \xi_\alpha(s) ,
\end{align}
together with $R \psi_\alpha(0) = R\varphi_0 \otimes \eta_0 = 0$, one finds by partial integration
\begin{subequations}
\begin{align}
f_\alpha (t)   &  = -  \alpha^{-1}    \lsp   \psi_\alpha(t )    ,     R   \phi ( \delta G_x  )  \xi_\alpha(t)  \rsp_{\mathscr H}   \label{eq: partial integration 01 line 1}\\[2mm]
&  + i \alpha^{-3}  \int_0^t    \lsp    \psi_\alpha(s)   ,     R   \big ( \big[N,\phi( \delta G_x  )\big] + \phi (\delta G_x ) A^{\varphi_0}  \big)        \xi_\alpha(s)   \rsp_{\mathscr H}\, \D s \label{eq: partial integration 01 line 2}\\[0mm]
&  + i \alpha^{-2}   \int_0^t     \lsp    \psi_\alpha(s)   ,      \phi( \delta G_x  )  R    \phi( \delta G_x ) \xi_\alpha(s)    \rsp_{\mathscr H}\, \D s. \label{eq: partial integration 01 line 3}
\end{align}
\end{subequations}
In the last line the prefactor $\alpha^{-2}$ is not sufficient and we need to do a second partial integration. For that, we insert again the identity $ 1 = P  + Q $ on the left of $\phi(\delta G_x )$. The term containing $P $ equals
\begin{align}\label{eq: second order P term}
 i \alpha^{-2}   \int_0^t    \lsp  \psi_\alpha(s)   ,     P \phi( \delta G_x )  R \phi ( \delta G_x )\xi_\alpha(s)  \rsp_{\mathscr H}\, \D s = - g_\alpha( t ) ,
\end{align}
and thus
\begin{align}
 \eqref{eq: partial integration 01 line 3} + g_\alpha(t)   =  i \alpha^{-2}   \int_0^t    \lsp  \psi_\alpha(s)  ,   Q  \phi( \delta G_x )  R  \phi ( \delta G_x  )  \xi_\alpha(s)  \rsp_{\mathscr H} \, \D s .
\end{align}
In this term we can integrate by parts similarly as in \eqref{eq: prepare for partial integration} which leads to
\begin{subequations}
\begin{align}
 \eqref{eq: partial integration 01 line 3}   +  g_\alpha(t)   & =  \alpha^{-2} \lsp \psi_\alpha( t )   ,       R  \phi( \delta G_x  )  R    \phi( \delta G_x   ) \xi_\alpha( t )  \rsp_{\mathscr H}  \label{eq: partial integration 02 line 1} \\[2mm]
& - i \alpha^{-4}  \int_0^t     \lsp \psi_\alpha(s)    ,     R    \big [N,\phi( \delta G_x  )  R   \phi( \delta G_x) \big]  \xi_\alpha(s)     \rsp_{\mathscr H}\,  \D s\label{eq: partial integration 02 line 2} \\[1mm]
& - i \alpha^{-4}  \int_0^t      \lsp \psi_\alpha(s)    ,      R      \phi( \delta G_x )   R   \phi( \delta G_x )  A^{\varphi_0} \xi_\alpha(s)      \rsp_{\mathscr H}\,  \D s \label{eq: partial integration 02 line 3} \\[1mm]
& - i \alpha^{-3}  \int_0^t      \lsp \psi_\alpha(s)  ,   P   \phi( \delta G_x  )  R    \phi( \delta G_x )    R    \phi( \delta G_x  )  \xi_\alpha(s)    \rsp_{\mathscr H}\, \D s  \label{eq: partial integration 02 line 4}\\[1mm]
& - i \alpha^{-3}  \int_0^t      \lsp \psi_\alpha(s)  ,   Q  \phi( \delta G_x  )  R    \phi( \delta G_x )    R    \phi( \delta G_x  )  \xi_\alpha(s)    \rsp_{\mathscr H}\, \D s . \label{eq: partial integration 02 line 5}
\end{align}
\end{subequations}
In the last line we do a third partial integration, i.e.
 \begin{subequations}
\begin{align}
\eqref{eq: partial integration 02 line 5} &  =  - \alpha^{-3} \lsp \psi_\alpha(t) ,  R \phi(\delta G_x  )  R \phi(\delta G_x  )  R  \phi(\delta G_x )  \xi_\alpha(t)  \rsp_{\mathscr H}   \label{eq: partial integration 03 line 1} \\[2mm]
& + i \alpha^{-5}  \int_0^t     \lsp  \psi_\alpha(s) ,   R  \big[N, \phi( \delta G_x  ) R \phi(\delta G_x  )R \phi( \delta G_x  ) \big]   \xi_\alpha(s) \rsp_{\mathscr H}\,  \D s\label{eq: partial integration 03 line 2} \\[1mm]
& +  i \alpha^{-5}  \int_0^t      \lsp \psi_\alpha(s)  ,   R  \phi(\delta G_x ) R  \phi(\delta G_x  )  R \phi(\delta G_x ) A^{\varphi_0}   \xi_\alpha(s)  \rsp_{\mathscr H}\,  \D s \label{eq: partial integration 03 line 3} \\[1mm]
& + i \alpha^{-4}  \int_0^t      \lsp \psi_\alpha(s),     \phi(\delta G_x ) R   \phi(\delta G_x ) R   \phi(\delta G_x  )  R    \phi(\delta G_x   )  \xi_\alpha(s) \rsp_{\mathscr H}\, \D s . \label{eq: partial integration 03 line 4}
\end{align}
\end{subequations}
\noindent Summing the above expansion up we arrive at
\begin{align*}
& f_\alpha(t)   + g_\alpha(t)  \nonumber\\[2mm]
& \quad =  \eqref{eq: partial integration 01 line 1} +  \eqref{eq: partial integration 01 line 2}  +  \eqref{eq: partial integration 02 line 1}  + \eqref{eq: partial integration 02 line 2}  + \eqref{eq: partial integration 02 line 3} + \eqref{eq: partial integration 02 line 4} +  \eqref{eq: partial integration 03 line 1}  +  \eqref{eq: partial integration 03 line 2}  +  \eqref{eq: partial integration 03 line 3} + \eqref{eq: partial integration 03 line 4}.
\end{align*}
In the remainder of the proof we separately estimate each summand on the right side. This is readily done using basic inequalities in combination with Lemmas \ref{Lemma: a priori estimates} and \ref{Lemma: a priori estimates 2}. At the end, we conclude by applying Gronwall's inequality.\medskip 

\noindent \textbf{Term \eqref{eq: partial integration 01 line 1}}. In the first boundary term from the partial integration we have
\begin{align}
\eqref{eq: partial integration 01 line 1} & = -  i  \alpha^{-1}    \lsp   \psi_\alpha(t)  -   \xi_\alpha(t) ,  R   \phi (\delta G_x  )    \xi_\alpha(t)  \rsp_{\mathscr H}
 \end{align}
since  $ R \xi_\alpha(t)  = 0$. Using the Cauchy--Schwarz inequality we obtain
\begin{align}
\vert \eqref{eq: partial integration 01 line 1}  \vert &   \le \frac{1}{4} \norm[\mathscr H]{   \psi_\alpha(t)  -   \xi_\alpha(t)  }^2 +  \alpha^{-2}  \norm[{\mathscr H}]{   R \phi( \delta G_x  )   \xi_\alpha(t)   }^2,
\end{align}
and with \eqref{eq: a priori bound 1} and \eqref{eq: time dependent a priori estimates 1},
\begin{align}
  \norm[{\mathscr H}]{   R \phi( \delta G_x)   \xi_\alpha(t)     }^2 \le   C \norm[\mathscr H]{  (N+1)^{1/2}    \xi_\alpha(t)    }^2 \le C  \exp(c \vert t \vert \alpha^{-2} ).
\end{align}
\textbf{Terms \eqref{eq: partial integration 02 line 1} and \eqref{eq: partial integration 03 line 1}}. For the other two boundary terms we proceed similarly and find  
\begin{align}
\vert \eqref{eq: partial integration 02 line 1} \vert  & \le  \frac{1}{4} \norm[\mathscr H]{  \psi_\alpha(t)  -   \xi_\alpha(t) }^2 + \alpha^{-4}   \norm[\mathscr H]{  R  \phi( \delta G_x  )  R \phi( \delta G_x   )    \xi_\alpha(t)  } \nonumber \\[2.5mm]
&\le \frac{1}{4} \norm[\mathscr H]{    \psi_\alpha(t)  -   \xi_\alpha(t)  }^2 +C\alpha^{ - 4} \exp(c\vert t\vert \alpha^{-2})
\end{align}
as well as
\begin{align}
\vert \eqref{eq: partial integration 03 line 1} \vert &  \le  \frac{1}{4} \norm[\mathscr H]{  \psi_\alpha(t)  -   \xi_\alpha(t) }^2 + \alpha^{-6}   \norm[\mathscr H]{R  \phi( \delta G_x  )    R  \phi( \delta G_x  )  R \phi( \delta G_x   )    \xi_\alpha(t)  } \nonumber \\[2.5mm]
&\le \frac{1}{4} \norm[\mathscr H]{    \psi_\alpha(t)  -   \xi_\alpha(t)  }^2 +C\alpha^{ - 6} \exp(c\vert t\vert \alpha^{-2})
\end{align} 
where we have used \eqref{eq: a priori bound 2} and \eqref{eq: a priori bound 3} in combination with \eqref{eq: time dependent a priori estimates 1}.\medskip

\noindent \textbf{Term \eqref{eq: partial integration 01 line 2}}. In this term we have
\begin{align}
\eqref{eq: partial integration 01 line 2} = i \alpha^{-3}  \int_0^t    \lsp     \psi_\alpha(s)  -   \xi_\alpha(s)   ,     R \big( \big[ N,\phi( \delta G_x  )\big ]  +  \phi (\delta G_x   ) A^{\varphi_0}  \big)    \xi_\alpha(s)  \rsp_{\mathscr H}\, \D s.
\end{align}
Using \eqref{eq: a priori bound 1}, the third line of \eqref{eq: a priori bound 3} and \eqref{eq: time dependent a priori estimates 1} we estimate
 \begin{align}
\vert \eqref{eq: partial integration 01 line 2}\vert & \le \alpha^{-2}  \int_0^t    \norm[\mathscr H]{   \psi_\alpha(s)  -   \xi_\alpha(s)   }^2 \, \D s  \nonumber \\[1.5mm]
& \quad  + \frac{1}{2} \alpha^{-4} \int_0^t \big( \norm[\mathscr H]{ R  \big[N,\phi(\delta G_x) \big]  \xi_\alpha(s)   }^2 + \norm[\mathscr H]{R \phi(\delta G_x ) A^{\varphi_0}  \xi_\alpha(s)  }^2 \big) \D s  \nonumber \\[1.5 mm]
& \le \alpha^{-2}     \int_0^t    \norm[\mathscr H]{    \psi_\alpha(s)  -   \xi_\alpha(s)     }^2 \, \D s  +  C \alpha^{-2} (\exp(c \vert t\vert \alpha^{-2})-1) .
\end{align}
\textbf{Terms \eqref{eq: partial integration 02 line 2} and \eqref{eq: partial integration 02 line 3}.} Similarly as in the previous term,
\begin{align}
& \vert \eqref{eq: partial integration 02 line 2} \vert   + \vert \eqref{eq: partial integration 02 line 3} \vert  \le \alpha^{-2} \int_0^t    \norm[\mathscr H]{   \psi_\alpha(s)  -   \xi_\alpha(s)  }^2\, \D s   \nonumber \\
& \quad +  \frac{1}{2}  \alpha^{-6}   \int_0^t  \big(   \norm[\mathscr H]{ R  \big[ N,  \phi(\delta G_x )  R  \phi(\delta G_x )  \big]   \xi_\alpha(s)    }^2 + \norm[\mathscr H]{ R  \phi(\delta G_x )  R  \phi(\delta G_x )  A^{\varphi_0}    \xi_\alpha(s)   }^2  \big) \D s  ,
\end{align}
and thus by means of \eqref{eq: a priori bound 2}, \eqref{eq: a priori bound 4} and \eqref{eq: time dependent a priori estimates 1} we obtain
\begin{align}
\vert \eqref{eq: partial integration 02 line 2} \vert   + \vert \eqref{eq: partial integration 02 line 3} \vert    \le \alpha^{-2} \int_0^t \norm[\mathscr H]{   \psi_\alpha(s)  -   \xi_\alpha(s)  }^2 \, \D s  +    C\alpha^{-4}  ( \exp(c\vert t \vert \alpha^{-2}) -1  ) .
\end{align}
\textbf{Term \eqref{eq: partial integration 02 line 4}.} In this line we keep the real part (cf.\ \eqref{eq: Duhamel expansion}) and have
\begin{align}
&  \re \eqref{eq: partial integration 02 line 4} = \alpha^{-3}  \int_0^t   \im  \lsp \psi_{\alpha}(s) - \xi_\alpha(s)   ,  P \phi(\delta G_x ) R  \phi(\delta G_x )  R   \phi(\delta G_x  )  \xi_\alpha(s)   \rsp_{\mathscr H}\,  \D s
\end{align}
(the imaginary part of the added expectation value is zero). The absolute value of the right side is bounded from above by
\begin{align}
\vert \re \eqref{eq: partial integration 02 line 4} \vert & \le \frac{1}{2}\alpha^{-2} \int_0^t  \big(  \norm[\mathscr H]{ \psi_{\alpha}(s) - \xi_\alpha(s) }^2  + \frac{1}{2} \alpha^{-4} \int_0^t    \norm[\mathscr H]{ P \phi(\delta G_x  ) R   \phi(\delta G_x  )  R   \phi(\delta G_x  )  \xi_\alpha(s)  }^2 \big) \D s  \nonumber \\
& \le \alpha^{-2} \int_0^t    \norm[\mathscr H]{   \psi_{\alpha}(s) - \xi_\alpha(s) }^2 \, \D s + C \alpha^{-2} ( \exp(c\vert t \vert \alpha^{-2}) - 1  ),
\end{align}
where one uses \eqref{eq: a priori bound 3} and \eqref{eq: time dependent a priori estimates 1} in the second step.\medskip

\noindent\textbf{Term \eqref{eq: partial integration 03 line 2}.} By means of \eqref{eq: a priori bound 3} and \eqref{eq: time dependent a priori estimates 1} one obtains
\begin{align}
\vert \eqref{eq: partial integration 03 line 2}  \vert  & \le \alpha^{-5}  \int_0^t   \norm{R \big [ N, \phi( \delta G_x  ) R \phi(\delta G_x  )R \phi( \delta G_x  ) \big ]  \xi_\alpha(s) }\, \D s \nonumber \\[2mm]
& \le C \alpha^{-3} (\exp(c\vert t \vert \alpha^{-2})-1).
\end{align}
\textbf{Term \eqref{eq: partial integration 03 line 3}.} In this term one can use \eqref{eq: a priori bound 5} and \eqref{eq: time dependent a priori estimates 1} to find
\begin{align}
\vert \eqref{eq: partial integration 03 line 3} \vert & \le \alpha^{-5}  \int_0^t     \norm{    R  \phi(\delta G_x  ) R   \phi(\delta G_x )  R \phi(\delta G_x ) A^{\varphi_0} \xi_\alpha (s)  }\, \D s  \le C \alpha^{-3} ( \exp(c\vert t \vert \alpha^{-2}) - 1) .
\end{align}
\textbf{Term \eqref{eq: partial integration 03 line 4}.} For the last term we apply \eqref{eq: a priori bound 6}
in combination with
\begin{align}
\snorm[\mathscr H]{  p \,  \psi_\alpha(s)   }  = \snorm[\mathscr H]{   p \, e^{ - i H^{\rm F}_ \alpha s}\varphi_0 \otimes W(\alpha f_0)^*  \eta_0 },
\end{align}
see \eqref{eq: com relation for HF}, as well as \eqref{eq: time dependent a priori estimates 1} and \eqref{eq: time dependent a priori estimates 2}. This leads to
\begin{align}
\vert \eqref{eq: partial integration 03 line 4} \vert & \le \alpha^{-4}  \int_0^t   \big\vert   \lsp \psi_\alpha(s)     \phi(\delta G_x ) R   \phi(\delta G_x ) R  \phi(\delta G_x )  R   \phi(\delta G_x )  \xi_\alpha(s)  \rsp_{\mathscr H}\big\vert\, \D s \nonumber \\
& \le C \alpha^{-4}  \int_0^t  \big(1+\norm[\mathscr H]{p \, e^{ - i H^{\rm F}_\alpha s}  \varphi_0 \otimes W(\alpha f_0)^*  \eta_0 } \big) \norm[\mathscr H]{ (N+1)^2 \xi_\alpha(s) }\, \D s \nonumber\\[2mm]
& \le C \alpha^{-2}( \exp( c \vert t \vert \alpha^{-2}) - 1) .
\end{align}
\textbf{Conclusion.} In total, we have shown
\begin{align}\label{eq: Groennwall inequality}
  \norm[\mathscr H]{ \psi_{\alpha}(s) - \xi_\alpha(s) }^2  & \le  C \alpha^{-2} \exp(c \vert t\vert \alpha^{-2} )    +  C \alpha^{-2}    \int_0^t  \norm[\mathscr H]{ \psi_{\alpha}(s) - \xi_\alpha(s) }^2 \, \D s,
\end{align}
from which the claimed bound follows by the integral version of Gronwall's inequality.\hfill $\square$\\

\noindent \label{sec: 2.3}2.3.\ \textbf{Proofs of Lemmas \ref{Lemma: a priori estimates} and \ref{Lemma: a priori estimates 2}.} \label{Sec: remaining proofs} The main tool of the proof of Lemma \ref{Lemma: a priori estimates} is the commutator method by Lieb and Yamazaki \cite{LiebY1958} by which one improves the behaviour of the interaction at large momenta using the regularity of the electron wave function. More precisely one writes
\begin{align}\label{eq: G as commutator}
G_x  = \tilde G_{x} -  p \cdot  K_{x}  + K_{x} \cdot p
\end{align}
with $\tilde G_x$ and $K_x$ defined by
\begin{align}
\tilde G_{x}(k) = G_x(k) \chi_{[0,1]}(\vert k \vert), \quad K_{ x}(k) = \frac{k}{\vert k \vert^2}G_x(k) \chi_{(1,\infty)}(\vert k \vert),
\end{align}
respectively, where $\chi$ denotes the characteristic function, i.e.\ $\chi_A(r)=1$ for all $r \in A\subseteq \mathbb R$ and $\chi_A(r)=0$ otherwise. The functions $\tilde G_{x}$ and $K_x$ are square-integrable,
\begin{align}\label{eq L2 norm of G and K}
\sup_{x\in \mathbb R^3}( \vert \vert \tilde G_{x} \vert \vert_{L^2} +  \vert \vert K_{ x}  \vert \vert_{L^2} ) <\infty,
\end{align}
and thus one can use the common bounds for the annihilation and creation operators, namely 
\begin{align}\label{eq: bounds for a a*}
\snorm[\mathscr H]{a(g) \Psi} \le  \snorm[L^2]{g}\, \snorm[\mathscr H]{N^{1/2}\Psi},\quad \snorm[\mathscr H]{a^*(g) \Psi}  \le  \snorm[L^2]{g}\, \snorm[\mathscr H]{(N+1)^{1/2}\Psi}
\end{align}
for any $g\in L^2(\mathbb R^3,\D k)$.
\begin{proof}[Proof of Lemma \ref{Lemma: a priori estimates}] For the proof of \eqref{eq: a priori bound 1}, we set $a^{\#} \in \{ a,a^*\}$ and use \eqref{eq: G as commutator}, \eqref{eq L2 norm of G and K}, $\snorm[L^2]{f_0}<\infty$ and \eqref{eq: bounds for a a*} to estimate
\begin{align}
\norm[\mathscr H]{   R   a^{\#} (\delta G_x  ) P \Psi } &\le   \norm[\mathscr H]{  R  a^{\#} (\tilde G_x -f_0 )  P  \Psi } +  \norm[\mathscr H]{  R   p \cdot a^{\#} (K_x)  P  \Psi}  +\norm[\mathscr H]{  R  a^{\#} ( K_x ) \cdot p P  \Psi } \nonumber \\[2mm]
& \le  C \big( \snorm[] {  R  }  + \snorm[]{  R  p} + \snorm[]{   R  } \, \snorm[]{p P } \big)  \norm[\mathcal F]{(N+1)^{1/2} \eta} \nonumber \\[2mm]
& \le C  \norm[\mathcal F ]{(N+1)^{1/2} \eta},\label{eq: bound for a star in the proof}
\end{align}
where $\snorm{ \cdot}  = \snorm[\mathscr L]{ \cdot}$ denotes the norm on the space of bounded operators $\mathscr L(L^2(\mathbb R^3,\D x))$. That $ \snorm{ R } + \snorm[L^2]{ p P} <\infty$ is clear. To show $\snorm{ R  p } <\infty$ we compute
\begin{align}
\norm[L^2]{p R  \psi}^2 & = \lsp  \psi, R   (h^{\varphi_0} - \lambda ) R  \psi \rsp_{L^2} +  \lsp \psi, R   ( \lambda - V^{\varphi_0}) R \psi \rsp_{L^2}\nonumber\\[2mm]
& \le \lsp  \psi, R  \psi \rsp_{L^2} +  \frac{1}{2} \lsp \psi,  R    p^2  R   \psi \rsp_{L^2} + C \norm[L^2]{R \psi}^2
\end{align}
where we used $\pm V^{\varphi_0}  \le \frac{1}{2} p^2 + C$ as shown, e.g.\ in  \cite[Lemma III.2]{LeopoldRSS2019}.\footnote{Note that our potential $ V^{\varphi_0}$ coincides (up to a factor) with $V_{\varphi}$ for $\varphi = f_0\in L^2(\mathbb R^3,\D k)$ in \cite{LeopoldRSS2019}.} Since the bound \eqref{eq: bound for a star in the proof} holds equally if $R$ is replaced by $P$ and since
\begin{align}\label{eq: commutator N phi}
\big[ N, \phi(\delta G_x  ) \big] = a^*(\delta G_x ) - a(\delta G_x ),
\end{align}
this proves \eqref{eq: a priori bound 1}.

In order to prove  \eqref{eq: a priori bound 2} we derive the bound for  $\norm{ R a^{\#_1}(\delta G_x  )  R   a^{\#_2}(\delta G_x)  \Psi  }$ with $a^{\#_i}  \in \{ a,a^*\}$. Proceeding similarly as in \eqref{eq: bound for a star in the proof}, we find 
\begin{align}
 \norm[\mathscr H]{ R a^{\#_1} (\delta G_x )  R    a^{\#_2}(\delta G_x )  \Psi  } &  \le C  \norm[\mathscr H]{ (N+1)^{1/2} R^{1/2} a^{\#_2}(\delta G_x )  \Psi  } . \label{eq: bound for a star a star in the proof}
\end{align}
From here we use
\begin{align}
(N+1)^{1/2} R^{1/2}  a(\delta G_x  )\Psi & = R^{1/2}  a(\delta G_x ) N^{1/2}\Psi,\\[1mm]
(N+1)^{1/2} R^{1/2}  a^*(\delta G_x  )\Psi & = R^{1/2}  a^*(\delta G_x ) (N+2)^{1/2}\Psi,
\end{align} 
together with
\begin{align}
\norm[\mathscr H]{   R^{1/2} a^{\#_2} (\delta G_x )   \Psi} \le C \norm[\mathcal F]{    (N+1)^{1/2}  \eta } .
\end{align}
The latter is obtained in complete analogy to \eqref{eq: bound for a star in the proof}. The bounds for the other terms on the l.h.s. of \eqref{eq: a priori bound 2} are derived the same way. Since the derivation of \eqref{eq: a priori bound 3} and \eqref{eq: a priori bound 4} is also very similar, we omit further details.

To prove \eqref{eq: a priori bound 5} we proceed again as in \eqref{eq: bound for a star in the proof} and find
\begin{align}
&  \big\vert \lsp \Phi, a^{\#_1} (\delta G_x  )R a^{\#_2}( \delta G_x ) R a^{\#_3}(\delta G_x  )  R a^{\#_4} (\delta G_x  ) \Psi \rsp_{\mathscr H}  \big\vert \nonumber \\[3mm]
& \le \norm[\mathscr H]{\Phi}\norm[\mathscr H]{ \big( a^{\#_1}( \tilde G_{x}-f_0 )  -  a^{\#_1}(K_{x} ) \cdot p \big) R   a^{\#_2} ( \delta G_x  ) R a^{\#_3}( \delta G_x ) R a^{\#_4} (\delta G_x  )   \Psi } \nonumber \\[3mm]
& \hspace{2cm}+ \Big\vert \lsp \Phi, p\cdot a^{\#_1}(K_{x}) R a^{\#_2} (\delta G_x ) R  a^{\#_3}( \delta G_x  ) R a^{\#_4} (\delta G_x  )  \Psi  \rsp_{\mathscr H}  \Big\vert  \nonumber \\[3mm] 
& \le C    \big(1 + \snorm[\mathscr H]{p\Phi} \big)  \norm[\mathscr H]{(N+1)^{1/2}   R^{1/2} a^{\#_2}( \delta G_x  ) R  a^{\#_3}(\delta G_x ) R^{1/2} a^{\#_4}( \delta G_x  ) \Psi}.
\end{align}
By estimating the last factor similarly as the right hand side of \eqref{eq: bound for a star a star in the proof} we obtain \eqref{eq: a priori bound 5}.
\end{proof}

\begin{proof}[Proof of Lemma \ref{Lemma: a priori estimates 2}] We start by verifying the following bound,
\begin{align}\label{eq: bound for time deriv of z(t)}
\big\vert \lsp \eta ,  (N+1)^{j-1} [N,A^{\varphi_0} ] (N+1)^{m-j} \eta \rsp_{\mathcal F}\big\vert  & \le C  \norm[\mathcal F]{ (N+1)^{m/2} \eta}^2
\end{align}
for $1\le j\le m $. To do so, use \eqref{eq: commutator N phi} to write
\begin{align}\label{eq: NA commutator}
 [N,A^{\varphi_0}] & = P  (a^*(  G_x) - a(G_x) ) R  (a^*(  G_x) + a(G_x) ) P  + \text{h.c.},
\end{align}
and then estimate each term separately. We illustrate the argument for the term $A^{\varphi_0}_{++} = P a^*(  G_x) R  a^*(  G_x)P $ for which we have
\begin{align}\label{eq: cases}
&  \lsp \eta ,  (N+1)^{j-1} A^{\varphi_0}_{++} (N+1)^{m-j} \eta \rsp_{\mathcal F}  \nonumber \\[2mm]
& = \begin{cases}
\lsp \eta , (N+1)^{j-1} (N-1)^{\frac{m}{2}+1-j} A^{\varphi_0}_{++}   (N+1)^{\frac{m}{2}-1} \eta \rsp_{\mathcal F}\quad \hspace{2.925cm} ( {\frac{m}{2}+1}\ge j ) ,\\[3mm]
\lsp \eta ,  (N+1)^{\frac{j+i}{2}-1} A^{\varphi_0}_{++}  (N+3)^{\frac{j-i}{2}} (N+1)^{m-j} \eta \rsp_{\mathcal F}\ \ \, \text{with}\   i  = m+2-j \ \ (  j\ge \frac{m}{2}+1 ).
\end{cases}\nonumber
\end{align}
Taking the absolute value and using the Cauchy--Schwarz inequality we can bound the first line from above by
\begin{align}
\norm[\mathcal F ]{(N+1)^{j-1} (N-1)^{\frac{m}{2}+1-j}\eta} \norm[\mathcal F]{A^{\varphi_0}_{++}  (N+1)^{\frac{m}{2}-1}  \eta}  \le C   \norm[\mathcal F]{ (N+1)^{\frac{m}{2}}\eta }^2 ,
\end{align}
where we used $\snorm[\mathcal F]{A^{\varphi_0}_{++} \eta }\le C \snorm[\mathcal F]{(N+1)\eta }$, $\eta \in \mathcal F$, which is proved the same way as the bound for the left side of \eqref{eq: bound for a star a star in the proof}. Similarly we find the following upper bound for the second line,
\begin{align}
\norm[\mathcal F]{ (N+1)^{\frac{j+i}{2}-1} \eta } \norm[\mathcal F]{A_{++}^{\varphi_0}  (N+3)^{\frac{j-i}{2}} (N+1)^{m-j}  \eta}  \le C   \norm[\mathcal F]{ (N+1)^{\frac{m}{2}}\eta}^2 .
\end{align}
Repeating the same argument for the other terms in \eqref{eq: NA commutator} leads to the stated bound in \eqref{eq: bound for time deriv of z(t)}. 

Next let $\eta_\alpha(t) = \exp( -i \alpha^{-2}(N-A^{\varphi_0 }) t )   \eta$ and compute the time-derivative
\begin{align}\label{eq: time deriv of z(t)}
 \frac{d}{dt} \lsp \eta_\alpha (t) , N^m \eta_\alpha (t)\rsp_{\mathcal F}   & = -  \alpha^{-2} \sum_{j=1}^m \lsp \eta_\alpha (t) ,  N^{j-1} i [N, A^{\varphi_0}] N^{m-j} \eta_\alpha (t)\rsp_{\mathcal F}  
\end{align}
which for $m \in \{1,2,3,4,5\}$ is easily checked explicitly. Setting $z(t) = \sum_{j=1}^5 \norm[{\mathcal F}]{(N+1)^{j/2} \eta_\alpha(t)}^2$, we have by \eqref{eq: bound for time deriv of z(t)} and \eqref{eq: time deriv of z(t)}, $\vert \frac{d}{dt} z(t)\vert \le C  \alpha^{-2} z(t) $. Since $\sup_{\alpha >0} z(0) <\infty$ by assumption, it follows from Gronwall's inequality that $z(t) \le C  \exp(c\vert t \vert \alpha^{-2}) $.


For a proof of \eqref{eq: time dependent a priori estimates 2}, let $\psi_\alpha(t) = e^{-i H^{\rm F}_\alpha t}  \varphi_0 \otimes  W(\alpha f_0)^* \eta $ and estimate 
\begin{align}
\norm[\mathscr H ]{  p\, \psi_\alpha (t) }^2 &  \le C  \lsp \psi_\alpha(t),  (H_\alpha^{\rm F} +1)  \psi_\alpha (t) \rsp_{\mathscr H}  = C  \lsp \psi_\alpha (0),  (H_\alpha^{\rm F} +1)  \psi_\alpha (0) \rsp_{\mathscr H} \nonumber \\[2mm]
& = C \big( 1 +  \mathcal E^{\rm P}(\varphi_0) +  \alpha^{-2} \lsp \eta, N   \eta \rsp_{\mathcal F} \big)
\end{align}
for some constant $C>0$. Here we used $N\ge 0$ and $p^2 + \alpha^{-2}N \le C ( H_\alpha^{\rm F}+1)$ in the first step (see e.g. \cite[Lemma A.5]{Griesemer2017}) and the commutation relation \eqref{eq: com relation for HF} together with  $\langle \varphi_0 , \delta G_x \varphi_0 \rangle_{L^2}=0$ in the third step.
\end{proof}

\noindent \label{sec: 2.4}2.4.\ \textbf{Proofs of Proposition \ref{prop: effective hamiltonian} and  Corollaries \ref{cor: reduced density} and \ref{cor: main estimate bogoliubov trafo}.} 
\begin{proof}[Proof of proposition \ref{prop: effective hamiltonian}] The identity in \eqref{eq: effective hamiltonian} follows from $\mathcal E^{\rm P}(\varphi_0) = \lambda + \snorm[L^2]{f_0}^2$ together with the commutation relations \eqref{eq: com relations 1} and \eqref{eq: com relations 2}. That $\mathscr D(N) \subseteq \mathscr D(N-A^{\varphi_0}) $ follows from
\begin{align}
\norm[\mathcal F]{ A^{\varphi_0}\eta }  = \norm[\mathscr H]{P\phi(G_x) R \phi(G_x) \varphi_0 \otimes \eta} \le C \norm[\mathcal F]{(N+1)\eta}
\end{align}
which is proven the same way as the bound for the l.h.s.\ of \eqref{eq: bound for a star a star in the proof}. Using \eqref{eq: commutator N phi} one further finds
\begin{align}
 \big \vert \lsp   \eta ,\big[ A^{\varphi_0} , N \big] \eta  \rsp_{\mathcal F} \big\vert \le C \lsp \eta , N  \eta \rsp_{\mathcal F}
\end{align}
for all $\eta \in \mathcal F_0$ with $\mathcal F_0\subseteq \mathcal F$ denoting the dense subspace of all Fock space vectors that have only finitely many nonzero components. Since $\mathcal F_0 $ is a core of the number operator $N$, we can infer that $N-A^{\varphi_0}$ is essentially self-adjointn by a variant of Nelson's commutator theorem \cite[Corollary 1.1]{Faris1974}. Alternatively one could conclude self-adjointness of $N-A^{\varphi_0}$ from the criteria for self-adjointness of Fock space operators found in \cite{Falconi2015}.
\end{proof}

In the following two proofs we make use of the bound
\begin{align}\label{eq: trace trace bound}
\textnormal{Tr}_{\mathscr H_1} \big\vert \textnormal{Tr}_{\mathscr H_2} \vert \Psi  \rangle  \langle \Phi \vert  \big\vert \le \snorm[\mathscr H_1 \otimes \mathscr H_2]{\Psi} \, \snorm[\mathscr H_1 \otimes \mathscr H_2]{\Phi}
\end{align}
where $\mathscr H_1,\mathscr H_2$ are two separable Hilbert spaces and $\Psi,\Phi \in \mathscr H_1\otimes \mathscr H_2$. The inequality follows from the variational characterization of the trace. For a proof see \cite[Appendix D]{FrankG2017}.

\begin{proof}[Proof of Corollary \ref{cor: reduced density}] Using 
\begin{align}
\vert \varphi_0 \rangle \langle \varphi_0 \vert   =  \textnormal{Tr}_{\mathcal F} \Big  \vert e^{-i H^{\varphi_0}_\alpha t }  \varphi_0 \otimes W(\alpha f_0)^* \eta_0  \Big \rangle \Big \langle e^{-i H^{\varphi_0}_\alpha t }  \varphi_0 \otimes W(\alpha f_0)^* \eta_0 \Big\vert 
\end{align}
in combination with \eqref{eq: trace trace bound} one readily finds
\begin{align} 
\textnormal{Tr}_{L^2}\Big\vert  \textnormal{Tr}_{\mathcal F}  \big\vert \Psi_{\alpha}(t) \big \rangle \big \langle \Psi_{\alpha}  (t) \big \vert    - \big  \vert \varphi_0 \big \rangle \big \langle \varphi_0 \big \vert   \Big\vert  \le 2 \norm[\mathscr H]{ \big( e^{-i H^{\rm F }_\alpha t } - e^{-i H^{\varphi_0}_\alpha t } \big) \Psi_{\alpha}(0) }.
\end{align}
Together with Theorem \ref{theorem: main theorem} this proves the corollary.\medskip
\end{proof}

\begin{proof}[Proof of Corollary \ref{cor: main estimate bogoliubov trafo}] Below we shall prove the identity
\begin{align}\label{eq: bogoliubov identity}
 \exp\big( -i \alpha^{-2}(N-A^{\varphi_0 } + \varepsilon )   t \big) \eta_0 =  U_{\mathcal V_\alpha(t)} \eta_0
\end{align}
where $\varepsilon = \int_{\mathbb R^3} (2\pi \vert k \vert)^{-2} \snorm[L^2]{ R^{1/2} e^{-ikx}\varphi_0}^2 \text{d}k $. With this identity at hand, we can proceed as in the proof of Corollary \ref{cor: reduced density}, i.e.\ we use
\begin{align}
& \big\vert U_{\mathcal V_\alpha(t)} \eta_0 \big \rangle \big \langle  U_{\mathcal V_\alpha(t)} \eta_0 \big\vert  = \textnormal{Tr}_{L^2} \Big\vert \varphi_0 \otimes U_{\mathcal V_\alpha(t)} \eta_0 \Big \rangle \Big \langle \varphi_0 \otimes  U_{\mathcal V_\alpha(t)} \eta_0 \Big\vert \nonumber \\[1mm]
& \hspace{0.5cm} = \textnormal{Tr}_{L^2} \Big\vert   \varphi_0 \otimes \exp(-i (N-A^{\varphi_0}) t )  \eta_0 \Big \rangle \Big \langle      \varphi_0 \otimes   \exp(-i (N-A^{\varphi_0}) t )   \eta_0   \Big\vert \nonumber \\[1mm]
& \hspace{0.5cm} = \textnormal{Tr}_{L^2} \Big\vert  W(\alpha f_0) e^{-i  H_\alpha^{\varphi_ 0}   t}   \varphi_0 \otimes W(\alpha f_0)^*\eta_0 \Big \rangle \Big \langle  W(\alpha f_0) e^{-i  H_\alpha^{\varphi_ 0}   t}    \varphi_0 \otimes   W(\alpha f_0)^* \eta_0 \Big\vert
\end{align}
and by means of \eqref{eq: trace trace bound} we thus obtain
\begin{align} 
& \textnormal{Tr}_{\mathcal F}\Big\vert  \textnormal{Tr}_{L^2 }  \big\vert W(\alpha f_0) \Psi_{\alpha}(t) \big \rangle \big \langle  W(\alpha f_0)  \Psi_{\alpha}  (t) \big \vert -  \big\vert U_{\mathcal V_\alpha(t)} \eta_0 \big \rangle \big \langle  U_{\mathcal V_\alpha(t)} \eta_0 \big\vert  \Big\vert  \nonumber \\[2mm]
& \hspace{7cm} \le 2 \norm[\mathscr H]{ \big( e^{-i H^{\rm F }_\alpha t } - e^{-i H^{\varphi_0}_\alpha t } \big) \Psi_\alpha(0) }.
\end{align}
\textit{Proof of \eqref{eq: bogoliubov identity}.} Here we follow the argument from \cite[Lem. 2.8 and App. B]{BossmannPPS2019} where a similar identity was proven in the context of the dynamics of weakly interacting bosons. The argument is based on some straightforward computations, well-known facts about Bogoliubov transformations and quasi-free states and a general result about the dynamics generated by quadratic Hamiltonians \cite[Prop. 7]{NamN2017}.

At this point it is useful to introduce the pointwise annihilation and creation operators $a_k,a_k^*$ defined by the requirement that 
\begin{align}\label{eq: pointwise creation operators}
a(g) = \int_{\mathbb R^3} \overline{g(k)}\, a_k\,  \text{d}k,  \quad a^*(g) = \int_{\mathbb R^3} \, g(k)\, a_k^*\,  \text{d}k 
\end{align}
for any $g\in L^2(\mathbb R^3, \text{d}k )$. The commutation relations \eqref{eq: canonical commutation relations} now read 
\begin{align}\label{eq: pointwise commutation relations}
[a_k,a_l^*]= \delta(k-l), \quad [a_k,a_l] = [a^*_k,a^*_l]= 0\ \ \ \forall\,  k,l \in \mathbb R^3.
\end{align}
Using \eqref{eq: pointwise creation operators} and \eqref{eq: pointwise commutation relations} and abbreviating $\varepsilon = \int_{\mathbb R^3} (2\pi \vert k \vert)^{-2} \snorm[L^2]{ R^{1/2} e^{-ikx}\varphi_0}^2 \text{d}k$ a short computation leads to
\begin{align}\label{eq: quadratic Hamiltonian}
 N - A^{\varphi_0} + \varepsilon =  \text{d}\Gamma (1 - \mathcal G) - \frac{1}{2}\int_{\mathbb R^3} \int_{\mathbb R^3} \big( \mathcal K(k,l) \,  a^*_k a_l^* + \overline{\mathcal K(k,l)} \, a_k a_l \big)  \text{d}k  \text{d}l
\end{align}
where $\text{d}\Gamma (1 - \mathcal G)$ denotes the second quantization of the one-body operator $1-\mathcal G$, i.e.
\begin{align}
\text{d}\Gamma (1 - \mathcal G) = \int_{\mathbb R^3} \,  a_k^* a_k   \, \text{d}k  -  \int_{\mathbb R^3}  \int_{\mathbb R^3} \mathcal G(k,l) \, a_k^* a_l\,  \text{d}k  \text{d}l,
\end{align}
see \eqref{eq: definition of K(k,l)} and \eqref{eq: definition of G(k,l)} for a definition of $\mathcal K(k,l)$ and $\mathcal G(k,l)$, respectively.
To the operator on the right side of \eqref{eq: quadratic Hamiltonian} we can apply \cite[Prop. 7]{NamN2017}. The requirements of this proposition are satisfied since $1 - \mathcal G : L^2(\mathbb R^3,\D k) \to L^2(\mathbb R^3, \D k)  $ is bounded and $\mathcal K:L^2(\mathbb R^3, \D k)  \to L^2(\mathbb R^3, \D k) $ is a Hilbert--Schmidt operator which can be verified by means of \eqref{eq: G as commutator}. By part (iii) of \cite[Prop. 7]{NamN2017} it follows in particular that for any quasi-free state $\eta_0 \in \mathcal F$, the time-evolved state $\eta_\alpha(t) = \exp(-i \alpha^{-2}  (N - A^{\varphi_0} +  \varepsilon)t )\eta_0$ is again quasi-free (the bound $\langle \eta_\alpha(t), N \eta_\alpha(t) \rangle \le C  \exp(c\vert t \vert \alpha^{-2})$ can be checked directly by means of Gronwall's inequality). It is further not difficult to verify that the state $U_{\mathcal V_\alpha(t)}\eta_0$ is also quasi-free ($\eta_0 = U_{\mathcal W}\Omega_0$ for some Bogoliuv map $\mathcal W$ and thus $U_{\mathcal V_\alpha(t)}\eta_0 = U_{\mathcal V_\alpha(t) \circ \mathcal W}\Omega_0$ with Bogoliubov map $\mathcal V_\alpha(t) \circ \mathcal W$). To show equality between the quasi-free states $\eta_\alpha(t)$ and $U_{\mathcal V_\alpha(t)}\eta_0$ we compare their reduced one-body density matrices. This is sufficient because of the well-known fact that quasi-free states are uniquely determined by their reduced one-body density matrices. For $\xi\in \mathcal F$ the reduced one-body density matrices  $\gamma_\xi   : L^2(\mathbb R^3, \D k) \to L^2( \mathbb R^3, \D k )$ and $\alpha_\xi : L^2(\mathbb R^3, \D k) \to L^2( \mathbb R^3, \D k )$ are defined by
\begin{align}
\lsp f, \gamma_\xi g\rsp_{L^2} = \lsp \xi , a^*(g)a(f) \xi \rsp_{\mathcal F}, \quad \lsp f, \alpha_\xi  \overline{g}\rsp_{L^2} = \lsp \xi , a(g)a(f) \xi \rsp_{\mathcal F}
\end{align}
for all $f,g\in L^2(\mathbb R^3, \D k)$. In order to show $ \gamma_{\eta_\alpha(t)} =  \gamma_{U_{{\mathcal V}_\alpha(t)} \eta_0 } $ and $ \alpha_{\eta_\alpha(t)} = \alpha_{U_{{\mathcal V}_\alpha(t)} \eta_0}$ we argue that they solve the same pair of differential equations with the same initial condition $\gamma_{\eta_0}$ and $\alpha_{\eta_0}$, respectively, and then use that the solution to this pair of differential equations is unique (the latter was shown in \cite[Prop. 7]{NamN2017}).

Instead of computing the time derivative of $\gamma_{\eta_\alpha(t)}$ and $\alpha_{\eta_\alpha(t)}$, and similarly for $U_{\mathcal V_\alpha(t)}\eta_0$ below, it is more convenient to determine the time derivative of $\lsp \eta_\alpha (t),	A(F_1) A(F_2) \eta_\alpha (t) 	\rsp_{\mathcal F} $ with $A(F)$ the generalized annihilation operator as defined above \eqref{eq: generalized annihilation operator}. For $F_1,F_2 \in L^2(\mathbb R^3, \D k)\oplus L^2(\mathbb R^3,\D k)$ we have
\begin{align}\label{eq: time-der of red dens 1a}
i\frac{d}{dt}\lsp \eta_\alpha (t),	A(F_1) A(F_2) \xi_\alpha (t) 	\rsp_{\mathcal F} & = \alpha^{-2}  \lsp \eta_\alpha (t) 	,	 \big[ N-A^{\varphi_0}, A(F_1) A(F_2) \big] \xi_\alpha (t) 	\rsp_{\mathcal F}
\end{align}
and it follows by a straightforward computation that
\begin{align}\label{eq: time-der of red dens 1b}
\big[ N - A^{\varphi_0}, A(F_1) A(F_2) \big]  =  A( \mathcal AF_1 ) A( F_2)  +  A( F_1) A( \mathcal A F_2) 
\end{align}
with
\begin{align}
\mathcal A =   \begin{pmatrix}
1-  \mathcal G   & \mathcal K \\ 
-\overline{ \mathcal K } &  -1 + \overline{ \mathcal G } 
\end{pmatrix}.
\end{align}
Next we use $U_ {\mathcal V}^* A(F) U_{\mathcal V} = A( \mathcal V^{-1} F)$, cf.\ \eqref{eq: unitary implementation}, to obtain
\begin{align}\label{eq: time-der of red dens 2a}
\lsp  U_{  {\mathcal V}_\alpha(t)} \eta_0 , 	A(F_1) A(F_2)  U_{  {\mathcal V}_\alpha(t)} \eta_0 	\rsp_{\mathcal F} & =  \lsp    \eta_0 , 	A(\mathcal V_\alpha^{-1 }(t) F_1) A( \mathcal V_\alpha^{-1 }(t) F_2)    \eta_0 	\rsp_{\mathcal F} .
\end{align}
By means of $ ( i\partial_t \mathcal V_\alpha^{-1}(t)) \mathcal V_\alpha(t) = -  \mathcal V_\alpha^{-1}(t) (i\partial_t \mathcal V_\alpha(t))$ together with $i\partial_t \mathcal V_\alpha(t) = \alpha^{-2} \mathcal A \mathcal V_\alpha(t)$, we can compute the time derivative
\begin{align}
& i \frac{d}{dt}  \lsp    \eta_0 , 	A(\mathcal V_\alpha^{-1 }(t) F_1) A( \mathcal V_\alpha^{-1 }(t) F_2)   \eta_0 	\rsp_{\mathcal F} \nonumber \\[2mm]
& =  \lsp    \eta_0 , 	\big( A(-i \partial_t \mathcal V_\alpha^{-1 }(t) F_1) A( \mathcal V_\alpha^{-1 }(t) F_2)    + A(\mathcal V_\alpha^{-1 }(t) F_1) A( -i \partial_t \mathcal V_\alpha^{-1 }(t) F_2)  \big)  \eta_0 	\rsp_{\mathcal F} \nonumber \\[2.5mm]
 & =  \alpha^{-2} \lsp    \eta_0 , 	\big( A(  \mathcal V_\alpha^{-1 }(t)\mathcal A F_1) A( \mathcal V_\alpha^{-1 }(t) F_2)    + A(\mathcal V_\alpha^{-1 }(t) F_1) A(  \mathcal V_\alpha^{-1 }(t) \mathcal A  F_2)  \big)  \eta_0 	\rsp_{\mathcal F} \nonumber \\[2,5mm]
 & =  \alpha^{-2}\lsp   U_{\mathcal V_\alpha(t)} \eta_0 , 	\big( A(   \mathcal A   F_1) A(   F_2)    + A( F_1) A(    \mathcal A F_2)  \big)  U_{\mathcal V_\alpha(t)} \eta_0 \rsp_{\mathcal F} . \label{eq: time-der of red dens 2b}
 \end{align}
Comparing \eqref{eq: time-der of red dens 1a} and \eqref{eq: time-der of red dens 1b} with \eqref{eq: time-der of red dens 2a} and \eqref{eq: time-der of red dens 2b} we see that the pairs of reduced one-body density matrices $(\gamma_{\eta_\alpha(t)} , \alpha_{\eta_\alpha(t)})$ and $(\gamma_{U_{{\mathcal V}_\alpha(t)} \eta_0 } , \alpha_{U_{{\mathcal V}_\alpha(t)} \eta_0})$ solve the same differential equation. Since the solution to this equation is unique, see \cite[Prop. 7]{NamN2017}, and since $\eta_\alpha(0) = U_{{\mathcal V}_\alpha(0)} \eta_0 = \eta_0$, we conclude their equality. This implies $\eta_\alpha(t) = U_{{\mathcal V}_\alpha(t)} \eta_0 $ and hence proves the claimed identity.
\end{proof}

\section*{Acknowledgements}

I thank Marcel Griesemer for many interesting discussions about the Fr\"ohlich polaron and also for valuable comments on this manuscript. Helpful discussions with Nikolai Leopold and Robert Seiringer are also gratefully acknowledged. This work was partially supported by the \emph{Deutsche  Forschungsgemeinschaft (DFG)} through the Research Training Group 1838: \emph{Spectral Theory and Dynamics of Quantum Systems}.

\vspace{5mm}
\noindent \textsc{Fachbereich  Mathematik, Universit\"at Stuttgart, Germany}\\[1mm]
\textit{E-mail address:} \texttt{mitrouskas@mathematik.uni-stuttgart.de}

\end{document}